\documentclass[11pt]{amsart}

\usepackage{mathptmx}
\usepackage{graphicx}
\usepackage{epsfig}
\usepackage{pstricks}

\newtheorem{thm}{Theorem}[section]

\newtheorem{lemma}[thm]{Lemma}

\newtheorem{remark}[thm]{Remark}
\newtheorem{example}[thm]{Example}

\usepackage{amsmath}
\usepackage{amsxtra}
\usepackage{amscd}
\usepackage{amsthm}
\usepackage{amsfonts}
\usepackage{amssymb}
\usepackage{eucal}

\newcommand{\cA}{\mathcal A}

\newcommand{\cM}{\mathcal M}

\newcommand{\Z}{{\mathbb Z}}

\newcommand{\Y}{{\mathbb Y}}

\newcommand{\ba}{{\mathbf a}}
\newcommand{\bb}{{\mathbf b}}
\newcommand{\bc}{{\mathbf c}}
\newcommand{\bk}{{\mathbf k}}
\newcommand{\bm}{{\mathbf m}}

\newcommand{\bx}{{\mathbf x}}
\newcommand{\by}{{\mathbf y}}
\newcommand{\bu}{{\mathbf u}}
\newcommand{\bd}{{\mathbf d}}
\newcommand{\bR}{{\mathbf R}}
\newcommand{\bC}{{\mathbf C}}
\newcommand{\bK}{{\mathbf K}}
\newcommand{\bL}{{\mathbf L}}
\newcommand{\bF}{{\mathbf F}}
\newcommand{\bG}{{\mathbf G}}
\newcommand{\bA}{{\mathbf A}}
\newcommand{\bJ}{{\mathbf J}}
\newcommand{\bz}{{\mathbf z}}
\newcommand{\bw}{{\mathbf w}}
\newcommand{\bone}{{\mathbf 1}}
\newcommand{\al}{{\alpha}}

\numberwithin{equation}{section}

\begin{document}
\title{Discrete integrable systems, positivity and continued fraction rearrangements}
\author{Philippe Di Francesco} 
\address{Department of Mathematics, University of Michigan,
530 Church Street, Ann Arbor, MI 48190, USA
and Institut de Physique Th\'eorique du Commissariat \`a l'Energie Atomique, 
Unit\'e de Recherche associ\'ee du CNRS,
CEA Saclay/IPhT/Bat 774, F-91191 Gif sur Yvette Cedex, 
FRANCE. e-mail: philippe.di-francesco@cea.fr}

\begin{abstract}
In this review article,
we present a unified approach to solving discrete, integrable, possibly non-commutative,
dynamical systems, including the $Q$- and $T$-systems based on $A_r$. 
The initial data of the systems are seen as cluster variables in 
a suitable cluster algebra, and may evolve by local mutations. 
We show that the solutions are always expressed as Laurent polynomials 
of the initial data with non-negative integer coefficients. This is done
by reformulating the mutations of initial data as local rearrangements
of continued fractions generating some particular solutions, that preserve
manifest positivity. We also show how these techniques apply as well to 
non-commutative settings.

\bigskip



\end{abstract}

\maketitle
\date{\today}

\section{Introduction}

\subsection{The recursions}

In these notes, we mainly deal with the so-called $Q$- and $T$-systems based on the Lie algebra $A_r$. These systems and their generalization to other Lie algebras were introduced \cite{KR} \cite{KNS} in the combinatorial study of the completeness of the Bethe 
Ansatz states for the diagonalization of generalized Heisenberg quantum spin chains. Apart from their original representation-theoretic
interpretation as recursion relations for Kirillov-Reshetikhin characters ($Q$-systems \cite{KR}) or $q$-characters ($T$-systems \cite{FR} \cite{Nakajima}), these equations have reappeared since in various combinatorial contexts \cite{AH} \cite{KT} \cite{SPY}. We use here harmlessly renormalized variables that allow, by introducing sign flips, 
to connect the systems to cluster algebras with trivial coefficients,
and by a slight abuse of language, 
we will still call the corresponding equations $Q$- and $T$-systems.

Let $I_r=\{1,2,...,r\}$,
we define the $A_r$ $Q$-system:
\begin{eqnarray}
R_{\al,k+1}R_{\al,k-1}&=&R_{\al,k}^2+R_{\al+1,k}R_{\al-1,k}  \qquad (\al\in I_r;k\in\Z)\label{Qsys} \\
R_{0,k}&=&R_{r+1,k}=1 \qquad\qquad\qquad\qquad (k\in\Z)
\end{eqnarray}
and the $A_r$ T-system:
\begin{eqnarray}
\quad T_{\al,j,k+1}T_{\al,j,k-1}&=&T_{\al,j+1,k}T_{\al,j-1,k}+T_{\al+1,j,k}T_{\al-1,j,k}\quad 
(\al\in I_r;j,k\in \Z)\label{Tsys} \\
T_{0,j,k}&=&T_{r+1,j,k}=1
\end{eqnarray}
This involves quantities $T_{\al,j,k}$ with fixed parity of $\al+j+k$ (which we take $=0$ mod 2).

The latter system is readily seen to reduce to the former, upon ``forgetting" the index $j$. One
way to realize this is to view the solutions of the $Q$-system as those of the $T$-system that are 2-periodic in 
the index $j$, with the correspondence $R_{\al,k}=T_{\al,\al+k\, {\rm mod}\, 2,k}$.

The $T$-system may also be viewed \cite{DFK09a} as a non-commutative generalization of the $Q$-system, 
in which $R_{\al,n}\to \bR_{\al,n}$ is no longer scalar, but an invertible operator acting on a Hilbert space with 
basis $\{|j\rangle\}_{j\in\Z}$.
In this formulation, $T_{\al,j,k}$ are the matrix elements of 
the operator $\bR_{\al,k}$ in the basis.
For $\al=0,1,...,r+1$ we define the actions:
\begin{equation} \label{actrt}
\bR_{\al,k} |j+k+\al\rangle = T_{\al,j,k} |j-k-\al\rangle \qquad (j,k\in\Z)
\end{equation}
Defining also the shift operator $\bd$ acting as $\bd |j\rangle=|j-1\rangle$ and its formal inverse $\bd^{-1}$,
the $T$-system is readily seen to be equivalent to the non-commutative $A_r$ Q-system:
\begin{eqnarray}
\bR_{\al,k+1}\,\bR_{\al,k}^{-1}\,\bR_{\al,k-1}&=&\bR_{\al,k}+\bR_{\al+1,k}\,\bR_{\al,k}^{-1}\,\bR_{\al-1,k} 
\qquad  (\al\in I_r;k\in\Z) \label{ncQsysone}\\
\bR_{0,k}&=&\bd^{2k} \qquad \qquad \qquad \qquad \qquad \qquad (k\in\Z)\\
\bR_{r+1,k}&=&\bd^{2r+2k+2} \qquad \qquad    \qquad \qquad \qquad (k\in\Z)
\end{eqnarray}

In the following, we will also consider the fully non-commutative version of \eqref{ncQsysone}
in the case of $A_1$, namely the recursion relation
\begin{equation}\label{ncqsys}
\bR_{n+1}\bR_n^{-1}\bR_{n-1}=\bR_n+\bR_n^{-1}
\end{equation}
for variables $\bR_n\in \cA$, a unital algebra, and
which corresponds to \eqref{ncQsysone} for $r=1$
but with different boundary conditions $\bR_{0,n}=\bR_{2,n}=\bf 1$, the unit of $\cA$.

\subsection{Initial data}

To fix entirely their solutions, we must supplement the systems above
with a suitable set of initial data.

\begin{figure}
\centering
\includegraphics[width=12.cm]{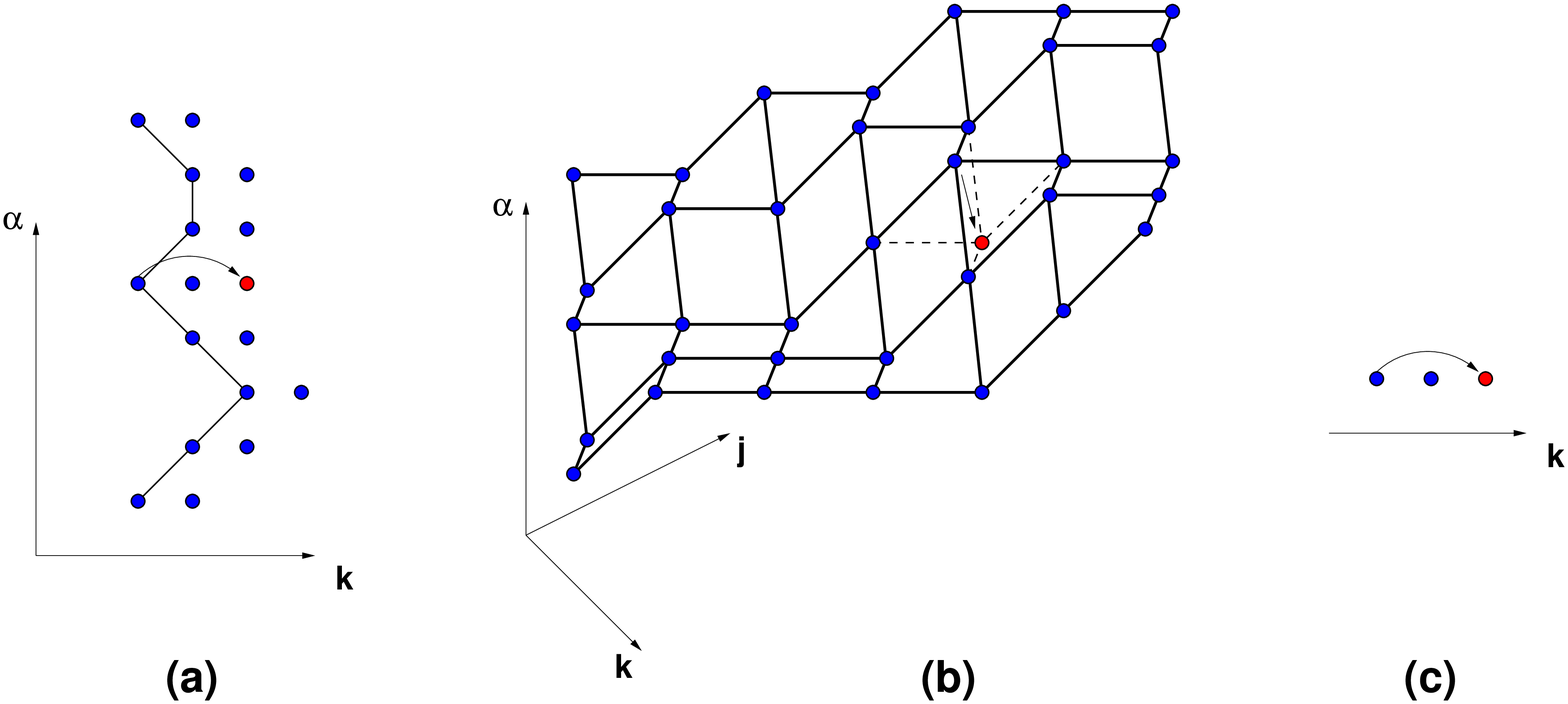}
\caption{\small Initial data structure for $Q$-systems (a), $T$-systems (b) and the non-commutative
$A_1$ $Q$-system (c). In all cases we indicate a sample forward mutation.}\label{fig:initial}
\end{figure}

For the $Q$-system, an obvious initial data set consists in fixing the values taken by
$\{R_{\al,m},R_{\al,m+1}\}_{\al\in I_r}$ for some fixed $m\in\Z$. Indeed, the 
system of recursion relations is a three-term relation in the variable $k$, hence fixing
all $R_{\al,k}$ for two consecutive values of $k$ allows to determine all other $R$'s.
Moreover the expression for $R_{\al,k}$ in terms of the 
$\{R_{\al,m},R_{\al,m+1}\}_{\al\in I_r}$ is the same as that for $R_{\al,k-m}$ in terms of
$\bx_0\equiv \{R_{\al,0},R_{\al,1}\}_{\al\in I_r}$, by translational invariance of the equations. We call
$\bx_0$ the fundamental initial data set. The most general admissible 
initial data sets are indexed by Motzkin paths $\bm=(m_1,m_2,...,m_r)\in\Z^r$ of length $r-1$,
with the property that $m_{i+1}-m_i\in \{0,1,-1\}$, and they read 
$\bx_\bm=\{R_{\al,m_\al},R_{\al,m_\al+1}\}_{\al\in I_r}$ (see Fig.\ref{fig:initial} (a)). 
For instance, the fundamental data set
$\bx_0=\bx_{\bm_0}$ is coded by the null Motzkin path $\bm_0=(0,0,...,0)$.
The initial data sets $\bx_\bm$ are related via
elementary moves, called mutations. We define the action of the 
(forward/backward) mutation $\mu_\al^{\pm}$, $\al\in I_r$
on a Motzkin path $\bm$ as $\mu_\al^{\pm}(\bm)=\bm'$, with $m_\beta'=m_\beta \pm \delta_{\al,\beta}$,
when $\bm'$ is also a Motzkin path (see Fig,\ref{fig:initial} (a) for an illustration). 
Assume $x_\bm$ is an admissible set of initial data,
then $x_{\bm'}$ differs from $x_\bm$ only by a substitution $R_{\al,m_\al}\to R_{\al,m_\al+2}$
(forward mutation $\mu_\al^+$) or
$R_{\al,m_\al+1}\to R_{\al,m_\al-1}$ (backward mutation $\mu_\al^-$). These substitutions
are done using the $Q$-system relation for $\al$, which involves only other terms
that are common to  $x_\bm$ and $x_{\bm'}$. It is easy to see that iterated mutations allow to connect 
the fundamental Motzkin path $\bm_0$ to any other Motzkin path $\bm$.

For the $T$-system, the obvious fundamental initial data set also consists of two layers with consecutive
values of $k$, hence $\bx_0=\{T_{\al,2j+\al,0},T_{\al,2j-1+\al,1}\}_{j\in \Z}$, up to translation. 
The most general admissible  initial data sets are indexed by ``stepped surfaces"
$\bk=(k_{\al,j})_{\al\in I_r;j\in \Z}$ with $k_{\al,j}\in \Z$ and the conditions that $k_{\al+1,j}-k_{\al,j}=\pm 1$
and $k_{\al,j+1}-k_{\al,j}=\pm 1$ (see Fig.\ref{fig:initial} (b)). 
The data set associated to the stepped surface $\bk$
is $x_\bk=\{T_{\al,j,k_{\al,j}}\}_{\al\in  I_r;j\in\Z}$.
The fundamental data set $\bx_0=\bx_{\bk_0}$ corresponds to the surface $\bk_0$ with
$k_{\al,j}=\al+j$ mod $2$. The initial data sets $\bx_\bk$ are related via
elementary moves, still called mutations. We define the action of the 
(forward/backward) mutation $\mu_{\al,j}^{\pm}$, $\al\in I_r$, $j\in \Z$,
on a stepped surface $\bk$ as $\mu_{\al,j}^{\pm}(\bk)=\bk'$, with 
$k'_{\beta,p}=k_{\beta,p} \pm 2\delta_{\al,\beta}\delta_{p,j}$,
when $\bk'$ is also a stepped surface (see Fig.\ref{fig:initial}(b) for an illustration). 
Assume $x_\bk$ is an admissible set of initial data,
then $x_{\bk'}$ differs from $x_\bk$ only by a substitution 
$T_{\al,j,k_{\al,j}}\to T_{\al,j,k_{\al,j}\pm 2}$. These substitutions
are done using the T-system relation for $\al$, which involves only other terms
that are common to  $x_\bk$ and $x_{\bk'}$. It is easy to see that iterated mutations allow to connect 
the fundamental stepped surface $\bk_0$ to any other stepped surface $\bk$.

For the non-commutative $A_1$ $Q$-system, admissible initial data have the form
$\bx_k=(\bR_k,\bR_{k+1})$, and (forward/backward) mutations lead to 
$\bx_{k\pm 1}$ (see Fig.\ref{fig:initial} (c)).

\subsection{Positivity and cluster algebra}

A common feature of all these systems is that they share the 
so-called positive Laurent property, namely that their solutions, when expressed in terms of any
of their admissible initial data, are Laurent polynomials with non-negative integer coefficients.

For the case of commuting variables this is actually part of a more general structure known
as cluster algebras \cite{FZI}. A rank $n$ cluster algebra (without coefficients)
is a dynamical system describing
the evolution of ``data" vectors $x_t=((x_t)_i)_{i=1}^n$ (clusters)
of $n$ formal variables (cluster variables), attached to the nodes of a Cayley tree
of degree $n$, with edges cyclically labeled $1,2,...,n$ around each node. This evolution 
is coded by a companion skew-symmetric $n\times n$ matrix $B_t$ with integer entries
called the exchange matrix, which evolves as well.
The pairs $(x_t,B_t)$ and $(x_u,B_u)$ at two adjacent nodes $t,u$ connected by an edge with label $k$
are related via the (involutive) mutation $\mu_k$, defined by $\mu_k(x_t,B_t)=(x_u,B_u)$ with:
\begin{eqnarray*}
(x_u)_i&=&\left\{ \begin{array}{ll} (x_t)_i & \hbox{if $i\neq k$} ;\\
{1\over (x_t)_k}\left\{\prod_{i=1}^n (x_t)_i^{[B_{k,i}]_+} +\prod_{i=1}^n (x_t)_i^{[B_{i,k}]_+}\right\}
& \hbox{otherwise}.\end{array}\right.\\
(B_u)_{i,j}&=&\left\{ \begin{array}{ll} -(B_t)_{i,j} & \hbox{if $i=k$ or $j=k$} ;\\
(B_t)_{i,j} + {\rm sign}((B_t)_{i,k})[(B_t)_{i,k}(B_t)_{k,j}]_+) & \hbox{otherwise}.\end{array}\right.
\end{eqnarray*}
Here we use the notation $[x]_+=$Max$(x,0)$. The exchange matrices $B$ are usually represented
pictorially as quivers, namely graphs with oriented possibly multiple
edges, with the condition that $B_{i,j}=$number of edges $i\to j$ when $B_{i,j}\geq 0$.
The finiteness restriction of rank $n$ may be removed, while keeping the same definitions,
leading to the notion of cluster algebras of infinite rank.
The main property of cluster algebras is the so-called Laurent property \cite{FZLaurent}\cite{FZI}: 
every cluster variable of any node $u$ may be expressed as a Laurent polynomial of the
cluster variables at any other node $t$. The positivity conjecture states that this 
polynomial always has non-negative integer coefficients \cite{FZI}, and
was only proved in a few cases so far 
(e.g. in the so-called acyclic cases \cite{POSIT} \cite{FZIV} \cite{ARS},
or that of clusters arising from surfaces \cite{MSW}).

The above $Q$- and $T$-systems have their variables belonging to specific cluster algebras
in the following way \cite{Ke07}\cite{DFK08}. 
The initial data sets turn out to be some particular clusters
of a cluster algebra of rank $2r$ ($A_r$ $Q$-system) or of infinite rank ($T$-system). To specify
the cluster algebras, we only need to give the exchange matrix $B_0$ at the fundamental
node $0$ of the Cayley tree, 
while $x_0=(R_{1,0},R_{2,0},...,R_{r,0},R_{1,1},...,R_{r,1})$ (Q-system) or 
$x_0=(T_{\al,j,\al+j\, {\rm mod}\,  2})_{\al\in I_r;j\in\Z}$ (T-system) 
is the suitably ordered fundamental set of initial data 
$\bx_{\bm_0}$ or $\bx_{\bk_0}$ defined above.  
For the $A_r$ $Q$-system, we have 
$B=\begin{pmatrix}0 & -C \\ C & 0\end{pmatrix}$ where $C$ is the 
Cartan matrix of $A_r$, $C_{a,b}=2\delta_{a,b}-\delta_{|a-b|,1}$.
For the $T$-system, $B_0$ is infinite. It has the entries:
$$(B_0)_{(\al,j),(\beta,k)}=(-1)^{\al+j}\left(\delta_{j,k}
\delta_{|\al-\beta|,1}-\delta_{\al,\beta}\delta_{|j-k|,1}\right)$$
with $\al,\beta\in I_r$ and $j,k\in\Z$.

\subsection{Results}

The aim of these notes is to present a unified approach for proving Laurent 
positivity in all three cases described above, and to extend it to more general cases.
This is essentially based on joint work with R. Kedem \cite{DFK3} \cite{DFK09a} \cite{DFK09b} \cite{NewRKPDF}.
The strategy consists of three steps. {\bf Step 1}: construct quantities conserved
modulo the evolution equations (discrete integrals of motion), and derive a linear recursion relation
for some subset of solutions, from which all solutions can be obtained. {\bf Step 2}: rephrase the linear recursion
relation in terms of a finite continued fraction expression for the corresponding generating function,
which is manifestly positive in the initial data. In principle steps 1 and 2
can be repeated for each initial data set, but in practice, it is more efficient to only carry them out for one
set of initial data (say the fundamental one), and then use the {\bf Step 3}: express mutations as
rearrangements of the initial continued fraction, that preserve the manifest positivity.

Note that continued fraction language immediately translates to that of weighted paths 
on some target graphs.
The Laurent positivity results may all be rephrased in terms of partition 
functions of positively weighted paths on suitable graphs. 
The mutations that implement changes of initial data
(rearrangements of the continued fractions) may be interpreted as transformations of both
the target graph and the weights.

This leads to Theorem \ref{positQ} for the $A_r$ Q-system (Section \ref{qsysec}), a restricted version of
Theorem \ref{posiTall} for the $A_r$ $T$-system (Section \ref{tsysec}) and Theorem
\ref{ncqpositaone} for the fully non-commutative $A_1$ $Q$-system (Section \ref{ncaonesec}).

This approach is actually restrictive inasmuch as $T$-system solutions are concerned. Indeed, it
only covers the case of translationally invariant stepped surfaces $\bk$, with $k_{\al,j+2}=k_{\al,j}$
(but arbitrary attached initial data $T_{\al,j,k_{\al,j}}$).
The only allowed mutations are actually compound, i.e. infinite iterations of elementary mutations
that map such surfaces
to each-other. To include the most general case of stepped surfaces, one must switch to a 
different point of  view, by constructing for each stepped surface a weighted graph or 
network, from which the solutions can be computed explicitly. We shall not cover this case
in the present notes, and refer to \cite{DFT} for details.

The advantage of the present approach however is that it has a natural generalization to fully
non-commutative settings, for which such notions as non-commutative continued
fraction rearrangements exist. We illustrate this further in Sections 
\ref{ncranktwosec} and \ref{ncaffinesec}. 
The general case corresponding to non-commutative weighted paths is 
investigated in Section  \ref{ncgensec}.

\medskip
\noindent{\bf Acknowledgments.} We would like to thank especially R. Kedem
for a fruitful collaboration. We also thank S. Fomin for hospitality 
at the Dept. of Mathematics of the University of Michigan, Ann Arbor, 
while this work was completed.
We received partial support from the ANR Grant GranMa, the
ENIGMA research training network MRTN-CT-2004-5652,
and the ESF program MISGAM. 

\section{$A_r$ $Q$-system}\label{qsysec}

\subsection{Warmup: the $A_1$ case}\label{warmupsec}
We start with the $A_1$ $Q$-system \eqref{Qsys} for $r=1$, 
with $R_n\equiv R_{1,n}$:
\begin{equation}\label{Qsysaone}R_{n+1}R_{n-1}=R_n^2+1
\end{equation}
We have:
\begin{thm}
The solution $R_n$ of the $A_1$ $Q$-system \eqref{Qsysaone}, when expressed in terms
of any admissible initial data $\bx_i=(R_i,R_{i+1})$, is a Laurent polynomial with
non-negative integer coefficients.
\end{thm}
\begin{proof}
We wish to express the general solution $R_n$ in terms of any admissible
initial data $\bx_i=(R_i,R_{i+1})$. Obvious symmetries of \eqref{Qsysaone} 
allow to restrict ourselves to $R_n$ for $n\geq 0$ and $\bx_0$ only. Indeed, 
introducing the two automorphisms $\sigma, \tau$
such that $\sigma(R_0)=R_1$ and $\sigma(R_1)=R_2$, 
while $\tau(R_0)=R_1$ and $\tau(R_1)=R_0$, we immediately see that $\sigma(R_n)=R_{n+1}$
and $\tau(R_n)=R_{1-n}$ for all $n\in \Z$, respectively from the invariances 
$n\to n+1$ and $n\to 1-n$ of the relation \eqref{Qsysaone}. If we have an expression 
$R_n=f_n(\bx_0)$, by applying to it $\sigma^i$ we get an expression 
$R_{n}=\sigma^i (f_{n-i}(\bx_0))=g_{n,i}(\bx_i)$ in terms of any $\bx_i$. Likewise, starting from
$R_n=f_n(\bx_0)$, $n\geq 0$, and applying $\tau$, we get $R_{1-n}=\tau(f_n(\bx_0))=g_n(\bx_0)$
which gives an expression for all $R_n$, $n<0$ in terms of $\bx_0$.
Moreover, if $f_n$, $n\geq 0$ is a Laurent polynomial with non-negative integer coefficients,
so are all the transformed expressions under iterations of $\sigma$ or $\tau$.

We now note that eq.\eqref{Qsysaone} expresses
that the ``discrete  Wronskian"
\begin{equation}
W_n=\left\vert \begin{matrix} R_{n+1} & R_n \\ R_n & R_{n-1} \end{matrix}\right\vert =1
\end{equation}
Therefore, we may write
$$W_{n+1}-W_n=0=\left\vert \begin{matrix} 
R_{n+2} +R_n & R_{n+1}+R_{n-1} \\ R_{n+1} & R_{n} \end{matrix}\right\vert$$
hence the columns of the resulting matrix must be proportional, and we write
$c_n={R_{n+2}+R_n\over R_{n+1}}={R_{n+1}+R_{n-1}\over R_n}=c_{n-1}=c$,
a constant independent of $n$. We compute:
\begin{equation}\label{consone}
c={R_{n+1}\over R_n}+{1\over R_n R_{n+1}}+{R_n\over R_{n+1}}
={R_1\over R_0}+{1\over R_1R_0}+{R_0\over R_1}
\end{equation}
This gives a conserved quantity (discrete first integral of motion) for the relation \eqref{Qsysaone}.
It may also be rephrased into a linear recursion relation for the $R_n$'s:
\begin{equation}\label{linrecone}
R_{n+1}-c R_n +R_{n-1}=0 \end{equation}
Introducing the generating function
$F_0(t)=\sum_{n=0}^\infty t^n \, R_n/R_0$, and defining the weights
$y_1=R_1/R_0$, $y_2=1/(R_0R_1)$ and $y_3=R_0/R_1=1/y_1$,
we may use \eqref{linrecone} to get:
\begin{eqnarray}
F_0(t)&=&\cfrac{1+({R_1\over R_0}-C)t}{1-ct +t^2}=
\cfrac{1-t(y_2+y_3)}{1-t(y_1+y_2+y_3)+y_1y_3t^2}\nonumber \\
&=&  \cfrac{1}{1-ty_1\cfrac{1-ty_3}{1-t(y_2+y_3)}}
=\cfrac{1}{ 1-t \cfrac{y_1}{1-t \cfrac{y_2}{1-t y_3}}}
\label{contifone}
\end{eqnarray}
The latter finite continued fraction expression has manifestly positive polynomial
coefficients of $y_1,y_2,y_3$ in its series expansion in $t$, and the Theorem follows
from the fact that $y_1,y_2,y_3$ are Laurent monomials in the initial data $\bx_0$.
\end{proof}

The continued fraction expression above may be interpreted as the
generating function for paths on the integer segment $[0,3]$, with steps $\pm 1$,
starting and ending at $0$ and with step weights $1$ for steps $i-1\to i$, $i=1,2,3$
and $t y_i$ for steps $i\to i-1$, $i=1,2,3$, each path being weighted by the product 
of its step weights.

The two possible forward mutations of the cluster algebra associated to \eqref{Qsysaone} read:
$\mu_1^+(R_{2i},R_{2i+1})=(R_{2i+2},R_{2i+1})$ and
$\mu_2^+(R_{2i},R_{2i-1})=(R_{2i},R_{2i+1})$, where the updates of initial data
are made by use of the relation \eqref{Qsysaone}. Let us express explicitly
the effect of the mutation $\mu_1^+$ on the continued fraction expression \eqref{contifone},
by noting the two following rearrangement Lemmas, both easily proved by direct calculation:

\begin{lemma}\label{reartone}
$$1+\cfrac{a}{1-a-b} =\cfrac{1}{1-\cfrac{a}{1-b}} $$
\end{lemma}

\begin{lemma}\label{reartwo}
$$a+\cfrac{b}{1-\cfrac{c}{1-u}} =\cfrac{a'}{ 1-\cfrac{b'}{ 1-c'-u}} $$
provided
$$a'=a+b\qquad b'=\cfrac{b c}{a+b} \qquad c'=\cfrac{ac}{a+b}$$
\end{lemma}

We write:
\begin{eqnarray*}
F_0(t)&=& \cfrac{1}{1-t \cfrac{y_1}{1-t \cfrac{y_2}{1-t y_3}}}=1+t y_1 F_1(t)\\
F_1(t)&=& \cfrac{1}{1-t y_1-t\cfrac{y_2}{1-ty_3}}=\cfrac{1}{1-t\cfrac{y_1'}{1-t\cfrac{y_2'}{1-ty_3'}}}
\end{eqnarray*}
where we have first applied Lemma \ref{reartone} and then Lemma \ref{reartwo}, with $u=0$ and:
\begin{eqnarray*}
y_1'&=&y_1+y_2={R_1^2+1\over R_0R_1}={R_2\over R_1}\\
y_2'&=&{y_2y_3\over y_1+y_2}={1\over R_1R_2}\\
y_3'&=&{y_1y_3\over y_1+y_2}={R_1\over R_2}
\end{eqnarray*}

We conclude that $F_1(t)$ is precisely the generating function $\sum_{n\geq 0} t^n R_{n+1}/R_1$
where the $R_n$'s are expressed in terms of $\bx_1=\mu_1^+(\bx_0)$. Note also that
$F_1(t)=\sigma(F_0(t))$, upon defining $\sigma(t)=t$. More generally, repeated
application of the rearrangement
Lemmas \ref{reartone} and \ref{reartwo} leads to iterated mutations.

We now have all the ingredients necessary for the generalization to the $A_r$ $Q$-system:
integrals of motion, continued fractions, rearrangements.

\subsection{The $A_r$ case}

The purpose of this section is to outline the proof the following:

\begin{thm}\label{positQ}
The solution $R_{\al,n}$ of the $A_r$ $Q$-system is a positive Laurent polynomial
of any admissible initial data $\bx_\bm$, for any Motzkin path $\bm$ of length $r-1$.
\end{thm}

As before, the symmetries of the system allow to restrict ourselves to $R_{\al,n}$ with $n\geq 0$ and to 
the Motzkin paths in a fundamental domain $\cM_r$ modulo global translation of all indices by $1$.

The $A_r$ $Q$-system \eqref{Qsys} possesses $r$ conserved quantities generalizing 
\eqref{consone}. These are obtained by repeating the argument of the $r=1$ case.

Introducing the discrete Wronskian 
\begin{equation}\label{discwron}
W_{\al,n}=\det_{1\leq a,b\leq \al} \big(R_{1,n-a-b+\al+1}\big)
\end{equation}
we have the relation $W_{\al,n+1}W_{\al,n-1}=W_{\al,n}^2+W_{\al+1,n}W_{\al-1,n}$
as a consequence of the Desnanot-Jacobi relation between the minors of any square matrix, 
also used by Dodgson for his famous condensation algorithm \cite{DOD}.
Moreover, we have $W_{1,\al}=R_{1,n}$ and $W_{0,n}=1$ by definition. 
We conclude that $W_{\al,n}=R_{\al,n}$, the solution of the $A_r$ $Q$-system. As a consequence,
we have the second ($A_r$-)boundary condition: $W_{r+1,n}=1$. Writing
$W_{r+1,n+1}-W_{r+1,n}=0$, we deduce the existence of a linear recursion relation
$$\sum_{a=0}^{r+1} (-1)^a c_a R_{n-a}=0$$
with $c_0=c_{r+1}=1$ and where $c_1,c_2,...,c_r$ are $r$ conserved 
quantities modulo \eqref{Qsys}. In Ref.\cite{DFK3}, these were expressed explicitly in terms
of the fundamental initial data $\bx_{\bm_0}$ corresponding to the null Motzkin path $\bm_0$,
as partition functions for hard particles on some target graph.

Introducing the generating function $F_{\bm_0}(t)=\sum t^n R_{1,n}/R_{1,0}$, the result of \cite{DFK3}
may be rephrased into the following continued fraction expression:
\begin{equation}\label{contifrazer}
F_{\bm_0}(t)=\cfrac{1}{1- t\cfrac{y_1}{1-t\cfrac{y_2}{1-ty_3-t\cfrac{y_4}{\ddots 
\cfrac{{}}{1-ty_{2r-1}-t\cfrac{y_{2r}}{
1-t y_{2r+1}}}}}}}
\end{equation}
where the weights $y_i\equiv y_i(\bm_0)$ are the following
explicit Laurent monomials of the initial data $\bx_{\bm_0}$:
\begin{equation}\label{ymzer}y_{2\al-1}={R_{\al,1}\, R_{\al-1,0}\over R_{\al,0}\, R_{\al-1,1}}
\qquad y_{2\al}={R_{\al+1,1}\, R_{\al-1,0}\over R_{\al,0} \, R_{\al,1}}
\end{equation}

\begin{remark}\label{rempath}
The continued fraction expression \eqref{contifrazer} has an immediate path interpretation.
We consider the paths on the integer segment $[0,r+2]$, from and to $0$, with steps $\pm 1,0$
with respective weights $1$ for steps $i-1\to i$, $t y_{2i}$ for steps $i+1\to i$ ($i=1,2,...,r$),
$t y_{2i-1}$ for steps $i\to i$ ($(i=2,3,...,r$), and $t y_1$ for the step $1\to 0$, $t y_{2r+1}$ for 
the step $r+2\to r+1$. Then $F_0(t)$ is the sum over all such paths
of the product of their step weights. Alternatively, this was interpreted in \cite{DFK3}
as the partition function for weighted paths on a graph $\Gamma_{\bm_0}$ with $2r+2$ vertices
labeled $0,1,2,...,r+2,2',3',...,r'$
and adjacency matrix $A$ with non-vanishing elements $A_{i,i+1}=A_{i+1,i}=1$ ($i=0,1,...,r+1$)
and $A_{i,i'}=A_{i',i}=1$ ($i=2,3,...,r$).
\end{remark}

\begin{figure}
\centering
\includegraphics[width=6.cm]{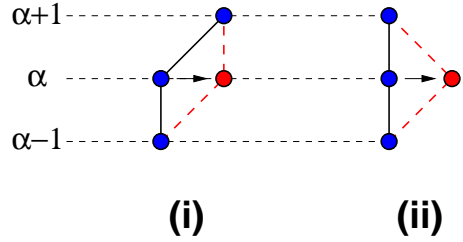}
\caption{\small The two cases (i) and (ii) of forward mutations $\mu_\al^+$ that allow to reach 
by iteration any Motzkin path, up to global translation, from the null path $\bm_0$.}\label{fig:twocases}
\end{figure}

The expression \eqref{contifrazer} has manifestly positive polynomial coefficients 
of the $y$'s in its series expansion in $t$, and the Theorem \ref{positQ} is proved for $R_{1,n}$
and the initial data $\bx_{\bm_0}$. To reach all the Motzkin paths in the fundamental
domain $\cM_r$, we proceed by induction under (forward) mutation. We note
that all the Motzkin paths in $\cM_r$ may be obtained from $\bm_0$ by iterated
application of mutations $\mu_\al^+:\bm \mapsto \bm'$, restricted 
only to the two following situations depicted in Fig.\ref{fig:twocases}
(where by convention we must simply omit $m_0$ and $m_{r+1}$):
\begin{itemize}
\item{\bf Case (i):} $m_{\al-1}=m_{\al}=m_{\al+1}-1$, then $m'_{\al-1}=m'_{\al}-1=m'_{\al+1}-1$.
\item{\bf Case (ii):}$m_{\al-1}=m_{\al}=m_{\al+1}$, then $m'_{\al-1}=m'_{\al}-1=m'_{\al+1}$.
\end{itemize}

In Ref.\cite{DFK3}, it was found that the mutations
may be implemented iteratively on the continued fraction $F_{\bm_0}(t)$ \eqref{contifrazer},
to reach any $F_\bm(t)=\sum_{n\geq 0} t^n R_{1,n+m_1}/R_{1,m_1}$, $\bm\in \cM_r$,
by application of the Lemmas \ref{reartone} and \ref{reartwo}, according to the case 
$\mu_\al^+:\bm\mapsto \bm'$ at hand:
\begin{itemize}
\item{$\al>1$:}  Apply Lemma \ref{reartwo} to $F_\bm(t)\to F_{\bm'}(t)$. 
\item{$\al=1$:} Apply first Lemma \ref{reartone} to obtain
$F_{\bm}(t)=1+t y_1(\bm) G_{\bm}(t)$, and then Lemma \ref{reartwo} to $G_{\bm}(t)\to F_{\bm'}(t)$.
\end{itemize}

Each continued fraction thus obtained is manifestly positive, 
as the two rearrangements are manifestly positive. Note that under the rearrangement
Lemma \ref{reartwo} the {\it structure} of the continued fraction changes, and in general
we obtain multiply branching continued fractions. Remarkably however, 
the continued fraction $F_{\bm}(t)$ may be expressed entirely in terms of a collection of weights
$\by(\bm)=(y_1(\bm),...,y_{2r+1}(\bm))$ called {\it skeleton weights}.
Alternatively, the continued fraction expressions allow to interpret the ratio
$R_{1,n+m_1}/R_{1,m_1}$ as the partition function for weighted paths 
on a rooted target graph $\Gamma_\bm$, from and to the root, and 
with exactly $n$ {\it descents}, i.e. steps taken towards the root. A compactified version
of these graphs was found recently in \cite{NewRKPDF}, and allows to put all the continued fractions $F_\bm$
in the form of Jacobi-type continued fractions, namely of the form
$$J(\bx)=\cfrac{1}{1-x_1 -\cfrac{x_2}{1- x_3 - \cfrac{x_4}{\ddots \cfrac{{} }{
1-x_{2r-1}- \cfrac{x_{2r}}{1- x_{2r+1}}}}}} $$
and we have
\begin{equation} \label{Jfrac}
F_\bm(t)=1+ty_1(\bm)J\big({\hat y}_1(\bm),...,{\hat y}_{2r+1}(\bm)\big)
\end{equation}
where for $\al=2,3,...,r$:
\begin{eqnarray}
&&\left\{ \begin{matrix}  
{\hat y}_{2\al-1}=ty_{2\al-1} \hfill & {\hat y}_{2\al}=ty_{2\al}  \hfill & 
{\rm if}\ m_{\al+1}=m_{\al} \hfill \\ 
{\hat y}_{2\al-1}=t(y_{2\al-1}+y_{2\al})  \hfill & {\hat y}_{2\al}=t^2y_{2\al}y_{2\al+1}  \hfill & 
{\rm if}\ m_{\al+1}=m_{\al}+1 \hfill \\
{\hat y}_{2\al-1}=ty_{2\al-1}-{y_{2\al}\over y_{2\al+1}}  \hfill & {\hat y}_{2\al}={y_{2\al}\over y_{2\al+1}}  \hfill 
& {\rm if}\ m_{\al+1}=m_{\al}-1  \hfill 
\end{matrix} \right.\nonumber \\
{\rm and} &&{\hat y}_{2r+1}=ty_{2r+1} \nonumber 
\end{eqnarray}
We see that the price to pay for having a nice Jacobi-type continued fraction structure for $F_\bm$ is that it is
no longer manifestly positive, due to the subtractions in the definition of $\hat y$. However, it is
a straightforward exercise to rearrange \eqref{Jfrac} into a manifestly positive, possibly multiply
branching continued fraction, in which all the coefficients are Laurent monomials of the $y_i(\bm)$'s with
coefficient $t$ (we call this form {\it canonical}). We
simply apply iteratively from the bottom to the top of the Jacobi-like continued fraction 
the Lemma \ref{reartone} in the form
\begin{equation}\label{flat}
a+b+\cfrac{bc}{1-c-u}=a+\cfrac{b}{1-\cfrac{c}{ 1-u}}
\end{equation}
or its ``converse":
\begin{lemma}\label{rearthree}
$$a-{b\over c} +\cfrac{{b\over c}}{1-c-u}=a+ \cfrac{b+{b\over c}u}{1-c-u}$$
\end{lemma}
Note that multiple branching will occur when applying Lemma \ref{rearthree}
because the remaining part $u$ of the continued fraction 
is now duplicated. Moreover, the iteration of the procedure will produce new weights,
which were called ``redundant weights" in \cite{DFK3}, 
that are all Laurent monomials of the skeleton weights.
Note finally that we must apply \eqref{flat} at ``level zero" to $F_\bm(t)$ itself, with $a=0$, $b=1$,
$c=y_1(\bm)$ (i.e. Lemma \ref{reartone}).
The manifestly positive canonical form of $F_\bm(t)$ can be directly interpreted \cite{DFK3}
as the partition function  of 
paths on a rooted target graph $\Gamma_\bm$, with step weights that are
either trivial (with value $1$) or 
Laurent monomials of the $y_i(\bm)$'s with coefficient $t$.

\begin{example}
The fraction corresponding to the Motzkin path $(2,1,0)$ for $r=3$ reads:
\begin{eqnarray*}
F_{2,1,0}(t)&=&1+t\cfrac{y_1}{1-ty_1+{y_2\over y_3} -\cfrac{{y_2\over y_3}}{1-ty_3
+{y_4\over y_5}-\cfrac{{y_4\over y_5}}{1-t y_5 -t \cfrac{y_6}{ 1-t y_7}}}}\\
&=&\cfrac{1}{1-t \cfrac{y_1}{1-t \cfrac{y_2 +\cfrac{{y_2y_4\over y_3}+\cfrac{{y_2y_4y_6\over y_3y_5}}{ 1-ty_7}
}{1-t y_5-t\cfrac{y_6}{ 1-t y_7}}}{
1-t y_3-t \cfrac{y_4+\cfrac{{y_4y_6\over y_5}}{1-t y_7}}{ 1-t y_5-t\cfrac{y_6}{ 1-t y_7}}}}}
\end{eqnarray*}
The first expression is the Jacobi-like form.
The second expression is the canonical form, obtained by applying Lemma \ref{rearthree} first on the third
level of the continued fraction, and then on the second, 
and by finally applying the Lemma \ref{reartone}.
It is manifestly positive, and introduces ``redundant" weights $ty_2y_4/y_3$, $ty_2y_4y_6/(y_3y_5)$,
$ty_4y_6/y_5$, all of which are Laurent monomials of the skeleton weights $(y_1,...,y_7)$
with coefficient $t$.
\end{example}
\begin{example}
The fraction corresponding to the Motzkin path $(0,1,0)$ for $r=3$ reads:
\begin{eqnarray*}
F_{0,1,0}(t)&=&1+t\cfrac{y_1}{ 1-t(y_1+y_2)-t^2 \cfrac{{y_2 y_3}}{ 1-ty_3
+{y_4\over y_5}-\cfrac{{y_4\over y_5}}{ 1-t y_5 -t \cfrac{y_6}{ 1-t y_7}}}}\\
&=&\cfrac{1}{1-t \cfrac{y_1}{1-t \cfrac{y_2}{ 
1-t \cfrac{y_3}{1-t \cfrac{y_4+\cfrac{{y_4y_6\over y_5}}{ 1-t y_7}}{ 1-t y_5-t\cfrac{y_6}{ 1-t y_7}}}}}}
\end{eqnarray*}
We have applied the Lemma \ref{rearthree} to the third stage of the initial continued fraction,
and then Lemma \ref{reartone} to its second and first stages. We get a redundant weight $ty_4y_6/y_5$.
\end{example}
Finally, we get the following explicit expression for the skeleton weights 
in terms of the initial data $\bx_\bm$, for any $\bm\in \cM_r$:

\begin{thm}
The skeleton weights associated to the motzkin path $\bm\in\cM_r$ read:
\begin{eqnarray}
y_{2\al-1}(\bm)&=&{R_{\al,m_\al+1}R_{\al-1,m_{\al-1}}
\over R_{\al,m_\al}R_{\al-1,m_{\al-1}+1}} \label{oddync} \\
y_{2\al}(\bm)&=&\left\{  \begin{matrix} 
{R_{\al+1,m_{\al+1}+1}R_{\al+1,m_{\al+1}-1}\over R_{\al+1,m_{\al+1}}^2} 
& {\rm if} \ m_\al=m_{\al+1}+1\\
{1} & {\rm otherwise} 
\end{matrix}\right\}  \nonumber \\
&&\qquad \times {R_{\al+1,m_\al+1}R_{\al-1,m_\al}\over R_{\al,m_\al+1}R_{\al,m_\al}} \nonumber \\
&&\qquad \times
\left\{ \begin{matrix} 
{R_{\al-1,m_{\al-1}}^2\over R_{\al-1,m_{\al-1}-1}R_{\al-1,m_{\al-1}+1} } & {\rm if} \ m_\al=m_{\al-1}-1\\
{1} & {\rm otherwise} 
\end{matrix}\right\}\label{evenync} 
\end{eqnarray}
\end{thm}
\begin{proof}
This is proved by induction on mutations. First, the formulas (\ref{oddync}-\ref{evenync})
hold for the fundamental case $\bm=\bm_0$ (see \eqref{ymzer}). Assume they hold
for some Motzkin path $\bm$.
As before, we only have to examine the cases (i) and (ii).
As a consequence of Lemma \ref{reartwo}, the mutation 
$\mu_\al:\by\equiv \by(\bm)\mapsto \by'\equiv\by(\bm')$ 
acts on the skeleton weights as follows:
\begin{eqnarray}
{\rm \bf Case}\  (i)\ :&& \left\{ \begin{matrix} 
y_{2\al-1}' & = & y_{2\al-1}+y_{2\al} \\
y_{2\al}' & = & y_{2\al+1}y_{2\al}/(y_{2\al-1}+y_{2\al}) \\
y_{2\al+1}' & = & y_{2\al+1}y_{2\al-1}/(y_{2\al-1}+y_{2\al})
\end{matrix} \right. \label{caseone}\\
{\rm \bf Case}\  (ii)\ :&&\left\{ \begin{matrix} 
y_{2\al-1}' & = & y_{2\al-1}+y_{2\al} \\
y_{2\al}' & = &  y_{2\al+1}y_{2\al}/(y_{2\al-1}+y_{2\al})  \\
y_{2\al+1}' & = & y_{2\al+1}y_{2\al-1}/(y_{2\al-1}+y_{2\al})\\
y_{2\al+2}' & = & y_{2\al+2}y_{2\al-1}/(y_{2\al-1}+y_{2\al})
\end{matrix} \right. \label{casetwo}
\end{eqnarray}
and $y'_\beta=y_\beta$ in all the other cases.
Finally, it is straightforward to check that
these relations are satisfied by the weights (\ref{oddync}-\ref{evenync}), 
as a consequence of the $Q$-system. The Theorem follows.
\end{proof}

Let us now consider $F_\bm(t)$ in its manifestly positive canonical form, with possible redundant weights,
which are all monomials of the skeleton weights, with coefficient $t$. 
The series expansion in $t$ has coefficients
that are polynomials of the skeleton and redundant weights, hence $R_{1,n+m_1}$
is itself a positive Laurent polynomial of the initial data $\bx_\bm$.
This completes the proof of Theorem \ref{positQ} for $\al=1$. The case of $\al>1$ 
follows from a slight generalization of the Lindstr\"om-Gessel-Viennot theorem \cite{LGV1}\cite{LGV2},
which allows to interpret the discrete Wronskian expression \eqref{discwron} for 
$R_{\al,n}=W_{\al,n}$ as the partition function for families of strongly non-intersecting paths on
the same target graph $\Gamma_\bm$, with the same step weights (see Ref.\cite{DFK3} for details).
The Theorem \ref{positQ} follows.

\section{$A_r$ $T$-system}\label{tsysec}

In this section, we present the solution of the $A_r$ $T$-system
when the boundary stepped surface is periodic.
It appears to be a non-commutative generalization of the $Q$-system case, and to have the Laurent positivity property. 
We actually have the following more general result, proved in \cite{DFT}:

\begin{thm}\label{posiTall}
For any stepped surface $\bk=(k_{\al,j})_{\al\in I_r;j\in\Z}$, 
the solution $T_{\al,j,k}$ of the $A_r$ $T$-system is a positive Laurent 
polynomial of the initial data $\bx_\bk$.
\end{thm}

\subsection{Periodic stepped surfaces}

In this section, we restrict ourselves to initial data $\bx_\bk$, with periodic stepped surfaces
$\bk$ with $k_{\al,j+2}=k_{\al,j}$ for all $\al,j$. 
An example is provided by the fundamental stepped
surface $\bk_0$.
These periodic stepped surfaces are related via compound mutations
of the form $\mu_\al^{\pm}=\prod_{j\in \Z} \mu_{\al,j}^{\pm}$, 
namely infinite products of elementary mutations,
which are implemented by the use of the $T$-system relations for a fixed 
$\al$ but for all $j\in\Z$ simultaneously, in order to simultaneously substitute 
$T_{\al,j,k_{\al,j\,{\rm mod}\, 2}}\to T_{\al,j,k_{\al,j\,{\rm mod}\, 2}\pm 2}$ for all $j\in\Z$.

The periodicity property allows us to work within the
framework of the non-commutative $Q$-system formulation of the $T$-system
(\ref{actrt}-\ref{ncQsysone}). The compound mutation $\mu_\al^{\pm}$ indeed 
amounts to a substitution $\bR_{\al,k}\to \bR_{\al,k\pm 2}$ by use of the
relation \eqref{ncQsysone}. 

Note that the operators $\bR_{\al,k}$ are only ``mildly" non-commuting, as
all the operators of the form $\bd^{-2\al-2k-m} \bR_{\al,k} \bd^m$ act diagonally on the
basis $|j\rangle_{j\in\Z}$, hence they all commute, for all $\al,k,m$.

\subsection{Non-commutative approach}

In Ref. \cite{DFK09a}, the constructions of the $Q$-system case are generalized
to the case of the $T$-system (conserved quantities, linear recursion relations, partition functions
for hard objects, and finally path formulation). 
In particular, the Desnanot-Jacobi identity may still be applied
to eliminate the $T_{\al,j,k}$ for $\al>1$ in terms
of the $T_{1,j,k}$'s via the following ``discrete Wronskian" formulas:
\begin{equation}\label{wronT}
T_{\al,j,k}=\det_{1\leq a,b\leq \al} \big( T_{1,j-a+b,k-a-b+\al+1} \big) 
\end{equation}
We may reformulate the final result for the solution $T_{1,j,k}$
of the $A_r$ $T$-system as a function of the fundamental initial data 
$\bx_{\bk_0}$ as follows. The generating function 
$\bF_0(t)=\sum_{k\geq 0} t^k \bR_{1,k} \bR_{1,0}^{-1}$ has a continued fraction
form analog to \eqref{contifrazer}.  Define the operators $\by_\al$, $\al=1,2,...,2r+1$
by:
\begin{eqnarray*}
\by_{2\al-1} |t\rangle &=& \ \
{T_{\al,t+\al-1,k}T_{\al-1,t+\al,0}\over T_{\al,t+\al-1}T_{\al-1,t+\al,1}} \ \ |t-2\rangle\\
\by_{2\al} |t\rangle &=& 
{T_{\al+1,t+\al,1}T_{\al-1,t+\al+1,0}\over T_{\al,t+\al,1}T_{\al,t+\al,0}} |t-2\rangle
\end{eqnarray*}
then we have:
\begin{eqnarray}
\bF_0(t)&=& \label{nconTi} \\
\left( I -t\big( I-t \big( I-t \by_3- \right. \!\!\!\!\!\!\!\!\!\!  &t&  \!\!\!\!\!\!\!\!\!\!\!\!  \left. \big(
\cdots (I-t\by_{2r-1}-t (I-t\by_{2r+1})^{-1}\by_{2r} )^{-1}\!\! 
\cdots\big)^{-1}\by_4 \big)^{-1}\by_2\big)^{-1}\by_1 \right)^{-1}\nonumber
\end{eqnarray}
where $I$ denotes the identity operator $I |j\rangle=|j\rangle$. 

Picking the coefficient of $|j-k-1\rangle$ in $\bF_0(t)|j+k-1\rangle$, 
this may be interpreted as a direct generalization 
of the $Q$-system result: $T_{1,j,k}/T_{1,j+k,0}$ is the partition function for 
paths on the same graph $\Gamma_{\bm_0}$ as in the commuting case
(see Remark \ref{rempath}), but now with 
``time"-dependent weights $y_\al(t)$ associated to the operators $\by_\al |t\rangle =y_\al(t)|t-2\rangle$,
such that the paths start at time $t_0=j-k-1$ and end at time $t_1=j+k-1$, 
and that the ``up" steps $+1$ take no time (i.e. do not increase the time counter) and have weight $1$, 
while the down steps $-1$ and level steps $0$ take two units of time 
(i.e. increase by $2$ the time counter) and receive the weights $y_i(t)$,
where $t$ is the value of the time counter after the step is completed. 
This shows explicit Laurent positivity for $T_{1,j,k}$ as a function of $\bx_{\bk_0}$.
(To make contact with \cite{DFK09a}, we have to introduce conjugated weights 
$\Y_\al=\bd^{-[{\al\over 2}]-1}\by_\al \bd^{[{\al\over 2}]+1}$ for $\al\geq 2$ and 
$\Y_1=\bd \by_1 \bd^{-1}$, with the effect of distributing the time more evenly on the steps,
namely steps $\pm1$ take one unit of time, while steps $0$ take two.)

Finally, let us show that, like in the commutative $Q$-system case, we may implement mutations by
non-commutative versions of the rearrangement Lemmas \ref{reartone} and \ref{reartwo}.
We have the following two non-commutative rearrangement of non-commutative fractions,
which preserve manifest positivity. 
(In the following, and for later generalizations,
we assume the continued fraction coefficients are elements of an algebra $\cA$ with unit $\bf 1$;
moreover, any expression of the form $({\bf 1}-\ba)^{-1}$
may be formally expanded as $\sum_{n\geq 0} \ba^n$.)

\begin{lemma}\label{ncreartone}
Let $\ba,\bb\in \cA$, we have:
$${\bf 1}+({\bf 1}-\ba-\bb)^{-1}\ba =\big({\bf 1}- ({\bf 1}-\bb)^{-1}\ba \big)^{-1} $$
\end{lemma}

\begin{lemma}\label{ncreartwo}
For $\ba,\bb,\bc,\bu\in \cA$, $\ba+\bb$ invertible, we have:
$$\ba+\big({\bf 1}-({\bf 1}-\bu)^{-1} \bc\big)^{-1}\bb 
=\big({\bf 1}-({\bf 1}-\bc'-\bu)^{-1}\bb'\big)^{-1}\ba'$$
provided
$$\ba'=\ba+\bb\qquad \bb'=\bc \bb (\ba+\bb)^{-1} \qquad \bc'=\bc \ba (\ba+\bb)^{-1}$$
\end{lemma}

We note that the periodic stepped surfaces $\bk$ may be indexed by 
Motzkin paths of length $r-1$. Indeed, the map 
$\varphi: \bk\mapsto \bm$ defined by $m_\al=$Min$(k_{\al,0},k_{\al,1})$ 
is a bijection between the set of periodic stepped surfaces and that of Motzkin 
paths of length $r-1$. By a slight abuse of notation, we
will denote by $\bF_\bm(t)$ the function $\bF_\bk(t)$ with $\bm=\varphi(\bk)$.
Starting from the operator continued fraction $\bF_{\bm_0}(t)$ \eqref{nconTi}, 
we may now apply iteratively the
above rearrangement Lemmas \ref{ncreartone} and \ref{ncreartwo}, 
exactly in the same way as we applied Lemmas \ref{reartone} and \ref{reartwo} 
for the commutative $Q$-system.
The results are operator continued fractions $\bF_{\bm}(t)$, with operator coefficients,
whose matrix elements in the basis $|j\rangle_{j\in\Z}$ are Laurent monomials of the 
initial data $\bx_\bk$ with coefficient $t$. As the rearrangements preserve positivity manifestly, we deduce
that the Theorem \ref{posiTall} holds for $\al=1$ and all periodic stepped surfaces $\bk$.
Alternatively, the rearranged continued fractions $\bF_\bm(t)$ generate time-dependent weighted
paths on the {\it same} target graphs $\Gamma_\bm$ as in the commuting $Q$-system case.

For $\al>1$, we interpret the discrete Wronskian formulas \eqref{wronT} 
via the same generalization of the Lindstr\"om-Gessel-Viennot theorem, as computing
the partition function for families of strongly non-intersecting paths on $\Gamma_\bm$,
now with time-dependent weights. The positivity and therefore Theorem \ref{posiTall} follows, 
for the case where the stepped surface $\bk$
is 2-periodic in the $j$ direction.

\section{Non-commutative systems}\label{noncosec}

We have already seen how to interpret the $T$-system as a non-commutative $Q$-system,
in terms of variables obeying many special commutation relations.
In this section we investigate integrable relations involving 
fully non-commutative variables.

\subsection{The non-commutative $A_1$ $Q$-system}\label{ncaonesec}

We have the following:
\begin{thm}\label{ncqpositaone}
The solutions of the non-commutative $A_1$ $Q$-system 
\eqref{ncqsys} in terms of any admissible initial data $\bx_i=(\bR_i,\bR_{i+1})$
are Laurent polynomials with non-negative integer coefficients.
\end{thm}

To prove the Theorem,
we must somehow repeat the steps of Section \ref{warmupsec}. 
As in the commuting case, the symmetries of the equation allow
us to restrict to the case $n\geq 0$ and to the initial data $\bx_0$. 
We still define anti-automorphisms
$\sigma,\tau$ that leave $\bf 1$ invariant, and such that 
$\varphi(a b)=\varphi(b)\varphi(a)$, $\varphi=\sigma,\tau$, and
$\sigma(\bR_0)=\bR_1$, $\sigma(\bR_1)=\bR_2$, 
while $\tau(\bR_0)=\bR_1$ and $\tau(\bR_1)=\bR_0$. We deduce that 
$\sigma(\bR_n)=\bR_{n+1}$ and $\tau(\bR_n)=\bR_{1-n}$ for all $n\in\Z$.
Assuming we have a positive Laurent polynomial expression for
$\bR_n(\bx_0)$ for all $n\geq 0$, we may obtain a positive Laurent expression for $n\leq 0$ as
well by applying $\tau$. Finally from a positive Laurent polynomial expression for
$\bR_n(\bx_0)$ for all $n\in \Z$, we may obtain one for 
$\bR_n(\bx_i)$ for any $i\in\Z$, by applying $\sigma^i$.

We now show the
existence of conserved quantities:

\begin{lemma}
The equation \eqref{ncqsys} has the two following conserved quantities:
\begin{eqnarray}
\bC&=&\bR_{n+1}^{-1}\bR_n\bR_{n+1}\bR_n^{-1}\label{nconsone}\\
\bK&=&(\bR_{n+1}+\bC \bR_{n-1})\bR_n^{-1}=\bR_n^{-1}(\bR_{n+1}\bC+\bR_{n-1})\label{nconstwo}
\end{eqnarray}
\end{lemma}
\begin{proof}
Introducing $\bC_n=\bR_{n+1}^{-1}\bR_n\bR_{n+1}\bR_n^{-1}$, 
and noting that $\bR_n$ commutes with the r.h.s. of \eqref{ncqsys}, we have
$\bR_n^2+1=\bR_{n+1}\bC_{n-1}\bR_{n-1}=\bR_{n+1}\bC_n \bR_{n-1}$, 
hence $\bC_n=\bC_{n-1}=\bC$. Note that the conservation of $\bC_n$ implies a
quasi-commutation relation:
\begin{equation}\label{quasicomm}
\bR_n\bR_{n+1}=\bR_{n+1}\bC \bR_n \end{equation}
and that the $A_1$ $Q$-system may be rewritten as:
\begin{equation}\label{neqaone}
\bR_{n+1}\bC\bR_{n-1}=\bR_n^2+\bf 1
\end{equation}
Introducing $\bK_n=(\bR_{n+1}+C \bR_{n-1})\bR_n^{-1}$
and $\bL_n=\bR_n^{-1}(\bR_{n+1}\bC+\bR_{n-1})$, 
we find that $\bR_n\bK_n\bR_n=\bR_n\bL_n\bR_n$
by use of the quasi-commutation \eqref{quasicomm}, hence $\bK_n=\bL_n$.
Moreover, computing 
$$\bR_{n+1}\bK_n\bR_n= \bR_{n+1}^2+\bR_{n}^2+{\bf 1}=\bR_{n+1}\bL_{n+1}\bR_n$$
we deduce that $\bK_n=\bL_{n+1}=\bK_{n+1}=\bK$, and the Lemma follows.
\end{proof}

The conserved quantities may be rewritten as:
\begin{eqnarray*}
\bC&=& \bR_{n+1}^{-1}\bR_n\bR_{n+1}\bR_n^{-1}=\by_3\by_1\\
\bK&=&\bR_{n+1}\bR_n^{-1}+\bR_{n+1}^{-1}\bR_n^{-1}+\bR_{n+1}^{-1}\bR_n
=\by_1+\by_2+\by_3
\end{eqnarray*}
where we have defined:
\begin{equation}\label{ncys}
\by_1=\bR_{1}\bR_0^{-1}\qquad \by_2=\bR_{1}^{-1}\bR_0^{-1}
\qquad \by_3=\bR_{1}^{-1}\bR_0
\end{equation}
The Lemma implies the existence of a linear recursion relation:
$$\bR_{n+1}-\bK \bR_n+\bC \bR_{n-1}=0$$

For $t$ a central element in $\cA$, introducing the formal generating function 
$\bF_0(t)=\sum_{n\geq 0}t^n \bR_n\bR_0^{-1}$, we easily compute:
\begin{eqnarray*}
\bF_0(t)&=&({\bf 1}-(\by_1+\by_2+\by_3) t+\by_3\by_1 t^2)^{-1}({\bf 1}-t(\by_2+\by_3))\\
&=&
\left({\bf 1} - t \left({\bf 1} - t \left({\bf 1} - t \by_3\right)^{-1}\by_2\right)^{-1} \by_1\right)^{-1}
\end{eqnarray*}
This may be expanded as a series in $t$, whose coefficients are polynomials of $\by_1,\by_2,\by_3$
with non-negative integer coefficients. The theorem \ref{ncqpositaone}
follows from the definition of $F_0(t)$ and the explicit form of the $\by$'s \eqref{ncys}.

Finally, let us show that, like in the $Q$ and $T$-system cases, we may implement mutations by
use of the non-commutative rearrangement Lemmas \ref{ncreartone} and \ref{ncreartwo}.
We write:
\begin{eqnarray*}
\bF_0(t)&=& \left({\bf 1}-t\big({\bf 1}-t({\bf 1}-t\by_3)^{-1}\by_2\big)^{-1}\by_1\right)^{-1}
={\bf 1}+t \bF_1(t)\by_1\\
\bF_1(t)&=& \big({\bf 1}-t \by_1-t({\bf 1}-t\by_3)^{-1}\by_2\big)^{-1}
=\left({\bf 1}-t\big({\bf 1}-({\bf 1}-\by_3')^{-1}\by_2'\big)^{-1}\by_1'\right)^{-1}
\end{eqnarray*}
where we have first applied Lemma \ref{ncreartone} and then Lemma \ref{ncreartwo}, with $\bu=0$ and:
\begin{eqnarray*}
\by_1'&=&\by_1+\by_2=(\bR_1+\bR_1^{-1})\bR_0^{-1}=\bR_2 \bR_1^{-1}\\
\by_2'&=&\by_3 \by_2 \bR_1 \bR_2^{-1}=\bR_1^{-1}\bC^{-1} \bR_2^{-1}=\bR_2^{-1}\bR_1^{-1}\\
\by_3'&=&\by_3\by_1 \bR_1 \bR_2^{-1}=\bC \bR_1 \bR_2^{-1}=\bR_2^{-1}\bR_1
\end{eqnarray*}
where we have used the $Q$-system relation \eqref{ncqsys}
(first line) and the quasi-commutation relation \eqref{quasicomm}
(second and third).

We conclude that $\bF_1(t)$ is precisely the generating function 
$\sum_{n\geq 0} t^n \bR_{n+1}\bR_1^{-1}$
where the $\bR_n$'s are now expressed in terms of $\bx_1=\mu_1^+(\bx_0)$. 
More generally, repeated
application of the non-commutative rearrangement
Lemmas \ref{ncreartone} and \ref{ncreartwo} leads to iterated mutations.

\subsection{The affine non-commutative rank two cases}\label{ncranktwosec}
The rank two affine cluster algebras have a
fundamental skew-symmetrizable exchange matrix $B_0=\begin{pmatrix}0 & -c \\
b & 0 \end{pmatrix}$, with $(b,c)=(2,2),(1,4)$ or $(4,1)$. 
The $A_1$ $Q$-system of previous section reduces to the case 
$b=c=2$ when $\bC={\bf 1}$. 
The two other cases also have non-commutative counterparts, 
namely the system
\begin{eqnarray}\label{qsysonefour}
\bR_{2n}\bR_{2n-1}^{-1} \bR_{2n-2}\bR_{2n-1}&=&1+\bR_{2n-1}\nonumber \\
\bR_{2n+1}\bR_{2n}^{-1}\bR_{2n-1}\bR_{2n}&=&1+(\bR_{2n})^4 
\end{eqnarray}
and that obtained by interchanging even and odd indices.

We have the positivity Theorem:
\begin{thm}\label{posionefour}
For all $n\in \Z$, the solutions $\bR_n$ of \eqref{qsysonefour} are positive Laurent
polynomials of any initial data $\bx_i=(\bR_i,\bR_{i+1})$.
\end{thm}

Repeating the approach of previous section, we note that the 
symmetries of the system allow to restrict to the case of $n\geq 0$ and 
to the initial data $\bx_0$ and $\bx_1$.
Next, we identify two conserved quantities:
\begin{eqnarray*}
\bC&=&\bR_{n+1}^{-1}\bR_n\bR_{n+1}\bR_{n}^{-1} \\
\bK&=&(\bR_{2n+2}+\bC\bR_{2n-2})\bR_{2n}^{-1}
\end{eqnarray*}
We finally arrive at the following continued fraction expression for the generating function
$\bF_0(t)=\sum_{n\geq 0} t^n \bR_{2n}\bR_0^{-1}$ in terms of $\bx_0$:
\begin{equation}\label{contfraconefour}
F(t)=({\bf 1}-\bK t+\bC t^2)^{-1}({\bf 1}-t(\bK-\bR_2\bR_0^{-1}))=
\left({\bf 1}-t \by_1 -t^2({\bf 1}-t \by_3)^{-1} \by_2\right)^{-1} 
\end{equation}
where
\begin{eqnarray*}
\by_1&=&\bR_2 \bR_0^{-1}=({\bf 1}+\bR_1^{-1})\bR_0^{-1}\bR_1\bR_0^{-1}\\
\by_2&=& (\bK-\by_1)\by_1 -\bC\\
&=&({\bf 1}+\bR_1^{-1})\bR_0^{-2}({\bf 1}+\bR_1^{-1})\bR_0^{-1}\bR_1\bR_0^{-1}
+\bR_1^{-1}\bR_0^2\bR_1^{-1}\bR_0^{-1}\bR_1\bR_0^{-1}\\
\by_3&=&\bK-y_1=\bC\bR_{-2}\bR_0^{-1}=({\bf 1}+\bR_1^{-1})\bR_0^{-2}
+\bR_1^{-1}\bR_0^2 
\end{eqnarray*}
This proves positivity of all $\bR_{2n}$ in terms of $\bx_0$. 
The case of $\bx_1$ is obtained by mutation/rearrangement of $\bF_0(t)$. We write:
\begin{eqnarray*}
t^{-1}(\bF_0(t)-{\bf 1})\by_1^{-1}&=&({\bf 1}-\bK t+\bC t^2)^{-1}({\bf 1}-t\bC\by_1^{-1})\\
&=& \left({\bf 1}-t \by_1' -t^2({\bf 1}-t \by_3')^{-1} \by_2'\right)^{-1} 
\end{eqnarray*}
where
\begin{eqnarray*}
\by_1'&=&\bK-\by_3'=\bR_4\bR_2^{-1}=\bR_2^{-1}(1+\bR_1)\bR_2^{-1}\bR_1^{-1}+\bR_2^2\bR_1^{-1} \\
\by_2'&=&\by_3'(\bK-\by_3')-\bC \\
&=& \bR_2^{-1}({\bf 1}+\bR_1)\bR_2^{-2} ({\bf 1}+\bR_1)\bR_2^{-1}\bR_1^{-1}+\bR_1^{-1}\\
\by_3'&=&\bC\by_1^{-1}=\bC \bR_0\bR_2^{-1}=\bR_2^{-1}({\bf 1}+\bR_1)\bR_2^{-1}
\end{eqnarray*}
We get a series in $t$ with coefficients that are polynomial in $\by_1',\by_2',\by_3'$, and
as each of them is itself a positive Laurent polynomial of the initial data  
we deduce positivity for all $\bR_{2n}$ in terms of $\bx_1$.
The case of the variables $\bR_{2n+1}=\bR_{2n+2}\bC \bR_{2n}-{\bf 1}$ is easy, by directly
checking that $\bR_{2n+2}\bC \bR_{2n}$ contains the term $\bf 1$. This completes the proof of
Theorem \ref{posionefour}.

\subsection{The non-commutative $\widehat{A_{2k}}$ systems}\label{ncaffinesec}

We now consider the systems
\begin{eqnarray}
\bu_{2n+2k+1} \bu_{2n}&=& \bu_{2n+1}\bu_{2n+2k}+{\bf 1}\label{rk3one}\\
\bu_{2n+1}\bu_{2n+2k+2}&=& \bu_{2n+2k+1}\bu_{2n+2} +{\bf 1}\label{rk3two}
\end{eqnarray}
with initial data $\bx_i=(\bu_i,\bu_{i+1},\ldots, \bu_{i+2k})$ for all $i\in\Z$.
The commutative version of (\ref{rk3one}-\ref{rk3two}) corresponds to mutations 
within the affine $\widehat{A_{2k}}$ cluster algebra.

We have the following positivity result (\cite{NewRKPDF}):
\begin{thm}\label{positrkthree}
The general solution $\bu_n$ to the system 
(\ref{rk3one}-\ref{rk3two}) is a positive Laurent
polynomial of any initial data $\bx_i$, for all $n,i\in\Z$.
\end{thm}

This is proved in the usual way. We may restrict ourselves to $n\geq 0$ and to the initial data 
$\bx_0$.
We find the conserved quantity
\begin{eqnarray*}
\bK&=&(\bu_{2n+2}+\bu_{2n-2})\bu_{2n}^{-1}=\bu_{2n+1}^{-1}(\bu_{2n+3}+\bu_{2n-1})\\
&=&\bu_0\bu_{2k}^{-1}+\bu_{2k}\bu_0^{-1}+\sum_{j=1}^k \bu_{2j-1}^{-1}(\bu_{2j}^{-1}+\bu_{2j+2}^{-1})
\end{eqnarray*}
This allows to derive continued fraction expressions for $\bF_i(t)=\sum_{n\geq 0} t^n \bu_{2i+2kn}$
and $\bG_i(t)=\sum_{n=0}^\infty t^n \bu_{2i+1+2kn}$,
$i=0,1,...,k-1$:
\begin{eqnarray*}
\bF_i(t)&=& ({\bf 1}-(\by_{2i}+\bz_{2i})t+t^2)^{-1}({\bf 1}-\bz_{2i} t) \bu_{2i}\\
&=&\left( {\bf 1}-\by_{2i} t-({\bf 1}-\bz_{2i}t)^{-1}\bw_{2i}t^2\right)^{-1}\bu_{2i}\\
\bG_i(t)&=& \bu_{2i+1} ({\bf 1}-\bz_{2i+1} t)({\bf 1}-(\by_{2i+1}+\bz_{2i+1})t+t^2)^{-1}\\
&=&\bu_{2i+1}\left( {\bf 1}-t\by_{2i+1}-t^2\bw_{2i+1}({\bf 1}-\bz_{2i+1} t)^{-1}\right)^{-1}
\end{eqnarray*}
where 
\begin{equation*}
\begin{matrix}
\by_{2i}=\bu_{2i+2k}\bu_{2i}^{-1}\hfill & 
\by_{2i+1}=\bu_{2i+1}^{-1}\bu_{2i+1+2k}\hfill \\
\bz_{2i}=\bK-\by_{2i}=\bu_{2i-2k}\bu_{2i}^{-1}\qquad \hfill & 
\bz_{2i+1}=\bK-\by_{2i+1}=\bu_{2i+1}^{-1}\bu_{2i+1-2k}\hfill \\
\bw_{2i}=\bz_{2i}\by_{2i}-{\bf 1}\hfill & \bw_{2i+1}=\by_{2i+1}\bz_{2i+1}-{\bf 1}\hfill 
\end{matrix}
\end{equation*}
We have manifest positivity in the $\by,\bz,\bw$'s. Moreover, 
the $\by$'s are defined recursively by $\by_0=\bu_{2k}\bu_0^{-1}$
and $\by_{2j}=\by_{2j-1}+\bu_{2j-1}^{-1}\bu_{2j}^{-1}$, $\by_{2j+1}=\by_{2j}+\bu_{2j+1}^{-1}\bu_{2j}^{-1}$,
hence they form subsums of $\bK$, and so do the $\bz$'s. Finally all $\by$'s contain the term
$\bu_{2k}\bu_0^{-1}$, while all $\bz$'s contain $\bu_0\bu_{2k}^{-1}$, hence both $\by_i\bz_i$ and $\bz_i\by_i$
contain $\bf 1$: we conclude that the $\by,\bz,\bw$'s are all positive Laurent polynomials of $\bx_0$.
This completes the proof of Theorem \ref{positrkthree}. 

\subsection{Non-commutative paths: a general study using quasi-determinants}\label{ncgensec}

In view of the accumulated evidence for path models and continued fractions 
to play an important role in the discrete non-commutative integrable systems,
we now address the general case of paths on the graphs $\Gamma_\bm$
of the $A_r$ $Q$-system solution, but with fully non-commutative weights.

For any $\bm$  Motzkin path of length $r-1$,
we define the family of non-commutative continued fractions 
$\bF_\bm(t)$
as follows. Let $(\by_1(\bm),...,\by_{2r+1}(\bm))\in\cA^{2r+1}$ be a given collection
of non-commutative ``skeleton" weights.
We define the non-commutative Jacobi-type fraction $\bJ(\bx)$ 
for a family of non-commutative weights
$\bx=(\bx_1,...,\bx_{2r+1})\in \cA^{2r+1}$ by:
\begin{equation}\label{jaconc}
\bJ(\bx)=
\left({\bf 1}-\bx_1 -\big( {\bf 1}-\bx_3 -\big(... ( {\bf 1}-\bx_{2r-1}-
({\bf 1}-\bx_{2r+1})^{-1}\bx_{2r})^{-1} ...\big)^{-1}\bx_4\big)^{-1}\bx_2\right)^{-1}
\end{equation}
We set
\begin{equation} \label{ncJfrac}
\bF_\bm(t)=1+t\, \bJ\big({\hat \by}_1(\bm),...,{\hat \by}_{2r+1}(\bm)\big)\, \by_1(\bm)
\end{equation}
where for $\al=2,3,...,r$:
\begin{eqnarray}
&&\left\{ \begin{matrix}  
{\hat \by}_{2\al-1}=t\by_{2\al-1}\hfill 
& {\hat \by}_{2\al}=t\by_{2\al}\hfill  
& {\rm if}\ m_{\al+1}=m_{\al}\hfill \\ 
{\hat \by}_{2\al-1}=t(\by_{2\al-1}+\by_{2\al})\hfill  
& {\hat \by}_{2\al}=t^2\by_{2\al+1}\by_{2\al}\hfill  
& {\rm if}\ m_{\al+1}=m_{\al}+1\hfill \\
{\hat \by}_{2\al-1}=t\by_{2\al-1}-\by_{2\al+1}^{-1}\by_{2\al}\hfill  
& {\hat \by}_{2\al}=\by_{2\al+1}^{-1}\by_{2\al}\hfill 
& {\rm if}\ m_{\al+1}=m_{\al}-1\hfill 
\end{matrix} \right.\nonumber \\
{\rm and} &&\quad {\hat \by}_{2r+1}=\by_{2r+1} \nonumber 
\end{eqnarray}

Like in the commutative case, we may easily bring the continued fraction $\bF_\bm(t)$ to
a manifestly positive form where all non-trivial coefficients are
Laurent monomials of the $\by$'s with coefficient $t$ (we still call this the
{\it canonical} form), 
and for which the coefficients of the formal series expansion in $t$
are therefore Laurent polynomials of the $\by$'s 
with non-negative integer coefficients. This is 
readily performed by applying iteratively from the bottom to the top of the 
Jacobi-type continued fraction
the non-commutative rearrangement Lemma
\ref{ncreartone}, in the form
\begin{equation}
\label{ncflat}
\ba+\bb+({\bf 1}-\bc -\bu)^{-1}\bc \bb=\ba+\big({\bf 1}-({\bf 1}-\bu)^{-1}\bc\big)^{-1}\bb
\end{equation}
or its ``converse":
\begin{lemma}\label{ncrearthree}
For all $\ba,\bb,\bc,\bu\in  \cA$, $c$ invertible, we have:
$$\ba-\bc^{-1} \bb +({\bf 1}-\bc-\bu)^{-1}\bc^{-1} \bb=\ba+( {\bf 1}-\bc-\bu)^{-1}(\bb+\bu\bc^{-1} \bb)$$
\end{lemma}
The net effect is to ``undo" the ${\hat \by}$'s, at the expense of creating 
branchings and new ``redundant" weights that are 
Laurent monomials of the skeleton weights, all with coefficient $t$, and therefore all manifestly positive.

The manifestly positive canonical form of $\bF_\bm(t)$ can now be directly interpreted
as the partition function  of non-commutative weighted paths on the {\it same} rooted 
target graph $\Gamma_\bm$ as for the commuting case.

The continued fractions $\bF_\bm(t)$ are related via non-commutative mutations as follows.
Restricting like in the commutative case to the cases (i) and (ii) for $\bm$ and 
$\bm'=\mu_\al^+(\bm)$, we see that the forward mutation $\mu_\al^+$ is implemented
on $\bF_\bm(t)$ written in the Jacobi form \eqref{ncJfrac} 
via the following sequences of rearrangements:
\begin{itemize} 
\item{\bf Case (i):}(a) Undo ${\hat \by}_{2\al-1},{\hat \by}_{2\al}$ 
by applying Lemma \ref{ncreartone}  (b) Apply Lemma \ref{ncreartwo} 
to the piece containing $\by_{2\al-1},\by_{2\al}$ and
$\by_{2\al+1}$ in factor (c) Go back to the Jacobi form
by using \eqref{ncflat} in reverse
\item{\bf Case (ii):}(a) Apply Lemma \ref{ncreartwo} to the piece containing 
$\by_{2\al-1},\by_{2\al},\by_{2\al+1},\by_{2\al+2}$ (b) Go back to the Jacobi form
by using \eqref{ncflat} and Lemma \ref{ncrearthree} in reverse
\end{itemize}
This produces $\bF_{\bm'}(t)$, in Jacobi form as well. Note that when $\al=1$ 
we have to apply \eqref{ncflat} at the level zero of the continued fraction,
with $\ba=0$, $\bb=\bone$ and $\bc=\by_1(\bm)$. So the transformation
of the continued fraction $\bF_\bm(t)$ reads:
\begin{eqnarray*}
\al=1:\quad \bF_\bm(t)&=&\bone +t  \bF_{\bm'}(t) \by_1(\bm)\\
\al>1:\quad \bF_\bm(t)&=& \bF_{\bm'}(t)\\
\end{eqnarray*}

It is easy to follow the transformations of 
skeleton weights in the process. 
The mutation $\mu_\al:\by\equiv \by(\bm)\mapsto \by'\equiv\by(\bm')$ acts on the skeleton weights as follows:
\begin{eqnarray}
{\rm \bf Case}\  (i)\ :&& \left\{ \begin{matrix} 
\by_{2\al-1}' & = & \by_{2\al-1}+\by_{2\al} \\
\by_{2\al}' & = & \by_{2\al+1}\by_{2\al} (\by_{2\al-1}+\by_{2\al})^{-1} \\
\by_{2\al+1}' & = & \by_{2\al+1}\by_{2\al-1} (\by_{2\al-1}+\by_{2\al})^{-1}
\end{matrix} \right. \label{caseonenc}\\
{\rm \bf Case}\  (ii)\ :&&\left\{ \begin{matrix} 
\by_{2\al-1}' & = & \by_{2\al-1}+\by_{2\al} \\
\by_{2\al}' & = & \by_{2\al+1}\by_{2\al} (\by_{2\al-1}+\by_{2\al})^{-1} \\
\by_{2\al+1}' & = & \by_{2\al+1}\by_{2\al-1} (\by_{2\al-1}+\by_{2\al})^{-1}\\
\by_{2\al+2}' & = & \by_{2\al+2}\by_{2\al-1} (\by_{2\al-1}+\by_{2\al})^{-1}
\end{matrix} \right. \label{casetwonc}
\end{eqnarray}
while all other skeleton weights are left unchanged.

We now wish to relate the skeleton weights $\by_\al(\bm)$ to the 
coefficients of the series expansion for $\bF_\bm(t)$. In particular, we
set $\bF_{\bm_0}(t)=\sum_{n\geq 0} t^n \bR_n\bR_0^{-1}$.

As shown in \cite{NewRKPDF} this can be done by the use of the theory of
quasi-determinants \cite{GGRW} \cite{GKLLRT}. 

For any square $k\times k$ matrix $\bA=(\ba_{i,j})$ with entries in $\cA$, 
and any $p,q\in\{1,2,...,k\}$, we define the $(p,q)$-quasi-determinant
$\vert \bA \vert_{p,q}$ as:
$$\vert \bA \vert_{p,q} =\ba_{p,q}- \sum_{i\neq p,j\neq q} \ba_{p,j} (\vert \bA \vert_{i,j})^{-1} \ba_{i,q}$$
whenever the r.h.s. is well-defined. The quasi-determinant has many interesting properties
(see  \cite{GGRW} for a detailed study). In the commutative case, it reduces to a ratio
of determinants, namely $\vert A \vert_{p,q}=\vert A \vert/\vert A^{p,q} \vert$, where $A^{p,q}$
stands for the matrix $A$ with row $p$ and column $q$ erased.
We now define the quasi-Wronskians of the sequence $(\bR_n)_{n\in \Z}$ to be:
\begin{equation} \label{wronskdefnc}
\Delta_{\al,n}=
\left\vert \begin{matrix} 
\bR_{n+\al-1} & \bR_{n+\al-2} & \cdots & \cdots & \bR_{n} \\
\bR_{n+\al-2} & \bR_{n+\al-3} & \cdots & \cdots & \bR_{n-1}\\
\vdots & \vdots &  & & \vdots \\
\bR_n & \bR_{n-1} & \cdots & \cdots & \bR_{n-\al+1}
\end{matrix} \right\vert_{1,1}
\end{equation}
with the convention $\Delta_{0,n}={\bf 1}$.
These quasi-Wronskians satisfy a generalized Desnanot-Jacobi relation
(also called the non-commutative Hirota equation \cite{NewRKPDF}):
\begin{equation}\label{nchiro}
\Delta_{\al+1,n}=\Delta_{\al,n+1}-\Delta_{\al,n}(\Delta_{\al,n-1}^{-1}-\Delta_{\al-1,n}^{-1})\Delta_{\al,n}
\end{equation}
The $A_r$ boundary condition is replaced by
$$ \Delta_{0,n}={\bf 1} \qquad \Delta_{r+2,n}=0 \quad (n\in \Z)$$
expressing simply that the continued fractions are actually
finite ($\by_{2r+2}(\bm_0)=0$), or in other words that there exists a vanishing 
left linear combination of the $\bR_n$'s with $r+2$ constant coefficients (the integrals of motion),
causing the quasi-determinant to vanish (this linear recursion relation for $\bR_n$, $n\geq 0$ is then used
to extend the definition of $\bR_n$ to $n<0$ as well, as the unique solution to the linear recursion relation).
Introducing new variables
\begin{equation}\label{defCD}
D_{\al,n}=\Delta_{\al+1,n}\Delta_{\al,n}^{-1} \qquad 
C_{\al,n}=\prod_{i=0}^{\al-1}\Delta_{\al-i,n-i}\Delta_{\al-i,n-i-1}^{-1}
\end{equation}
we may recast the above into the relations:
\begin{eqnarray}
D_{\al,n+1}&=& C_{\al+1,n+1}C_{\al,n}^{-1} D_{\al,n} C_{\al-1,n}C_{\al,n+1}^{-1}\label{done}\\
C_{\al,n+1}C_{\al-1,n}^{-1}&=&D_{\al,n}+C_{\al,n}C_{\al-1,n}^{-1}\label{dtwo}\\ 
C_{\al+1,n}C_{\al,n-1}^{-1}&=&D_{\al,n}+C_{\al+1,n}C_{\al,n}^{-1}\label{dthree}
\end{eqnarray}

Using these relations, we get the non-commutative version of (\ref{oddync}-\ref{evenync}):
\begin{thm}\label{ncweights}
The skeleton weights for $\bF_\bm(t)$ are expressed in terms of the $\bR_n$'s as:
\begin{eqnarray}\qquad
\by_{2\al-1}(\bm)&=&C_{\al,m_\al+1}C_{\al-1,m_{\al-1}+1}^{-1}\label{ncodd} \\
\by_{2\al}(\bm)&=&\left\{  \begin{matrix} 
C_{\al+1,m_{\al+1}+1}C_{\al+1,m_{\al}+1}^{-1} & {\rm if} \ m_\al=m_{\al+1}+1\\
{\bf 1} & {\rm otherwise} 
\end{matrix}\right\}\nonumber \\
&\times&
D_{\al,m_\al+1} \times \left\{ \begin{matrix} 
C_{\al-1,m_{\al}+1}C_{\al-1,m_{\al-1}+1}^{-1} & {\rm if} \ m_\al=m_{\al-1}-1\\
{\bf 1} & {\rm otherwise} 
\end{matrix}\right\}\label{nceven} 
\end{eqnarray}
\end{thm}

\begin{example}\label{fundstiel}
The weights for the fundamental Motzkin path $\bm_0=(0,...,0)$ read:
$$\by_{2\al-1}(\bm)=C_{\al,1}C_{\al-1,1}^{-1}\qquad \by_{2\al}(\bm)=D_{\al,1} $$
Those for the ``ascending" Motzkin path $\bm_1=(0,1,2,...,r-1)$ are:
\begin{eqnarray*}
\by_{2\al-1}(\bm_1)&=&C_{\al,\al}C_{\al-1,\al-1}^{-1}=\Delta_{\al,\al}\Delta_{\al,\al-1}^{-1}\\
\qquad \by_{2\al}(\bm_1)&=&D_{\al,\al}=\Delta_{\al+1,\al}\Delta_{\al,\al}^{-1}
\end{eqnarray*}
This is nothing but the non-commutative version \cite{GKLLRT} of the Stieltjes formula
\cite{STIEL1} expressing the coefficients of the ordinary continued fraction expansion in terms of
those of the series expansion. In a sense, our construction is a double generalization of that formula.
\end{example}

Note that, due to the explicit value of 
$\by_1(\bm)=C_{1,m_1+1}=\bR_{1,m_1+1}\bR_{1,m_1}^{-1}$, we have
$\bF_\bm(t)=\sum_{n\geq 0} t^n \bR_{1,n+m_1}\bR_{1,m_1}^{-1}$.
We conclude that for all Motzkin paths $\bm$ of length $r-1$,
the quantity $\bR_{n+m_1}\bR_{m_1}^{-1}$ is the partition function
for non-commutative weighted paths on $\Gamma_\bm$. All the weights
are monic Laurent monomials of the skeleton weights of Theorem \ref{ncweights},
hence the result is a positive Laurent polynomial of the variables $C,D$
occurring in Theorem \ref{ncweights}. Note that these are themselves 
Laurent monomials of certain quasi-Wronskians \eqref{wronskdefnc}.

Except for the point $\bm_1$ of Example \ref{fundstiel} above,
for which we have a clear Laurent positive phenomenon in
the variables $(\Delta_{\al,\al},\Delta_{\al,\al-1})_{\al=1}^{r+1}$,
the classification of relevant initial data in terms of the variables 
$\Delta_{\al,n}$ for the equation \eqref{nchiro} is not clear at this stage. 
All we have at our disposal is Laurent positivity in terms of the
non-commutative variables $C,D$.
A better understanding of what the ``good" dynamical variables are
should lead to a notion of non-commutative cluster algebra, via the 
mutations thereof.

Apart from the cases already treated in these notes, all of which may be
viewed as particular cases of this section, an important application concerns
the so-called quantum cluster algebras \cite{BZ}. These are non-commutative
versions of cluster algebras, in which the variables obey
$q$-commutation relations. 
A consequence of our construction
is that the positive Laurent phenomenon holds for all $\bR_n\equiv \bR_{1,n}$, $n\geq 0$ 
for the corresponding $A_r$ quantum $Q$-system \cite{NewRKPDF}:
$$q^{\al(r+1-\al)} \, \bR_{\al,n+1}\bR_{\al,n-1} =(\bR_{\al,n})^2+\bR_{\al+1,n}\bR_{\al-1,n}$$
with the boundary condition $\bR_{0,n}=\bR_{r+1,n}=\bf 1$, and 
where the variables within the same cluster (or initial data set) have the $q$-commutation
relations for $1\leq \al<\beta \leq r+1$:
$$ \bR_{\al,n} \, \bR_{\beta,n+p} =q^{\al(r+1-\beta)p} \, \bR_{\beta,n+p}\, \bR_{\al,n}$$

\begin{thm}
As a function of the initial data $\bx_\bm=\{\bR_{\al,m_\al},\bR_{\al,m_\al+1}\}_{\al\in I_r}$,
the quantity  $\bR_{1,n+m_1}\bR_{1,m_1}^{-1}$ is a Laurent polynomial with coefficients
in $\Z_+[q,q^{-1}]$.
\end{thm}

\begin{proof}
We simply have to compute the quasi-Wronskians:
$$\Delta_{\al,n}=q^{-{1\over 2}(\al-1)(2r+2-\al)}\, \bR_{\al,n}\bR_{\al-1,n-1}^{-1} $$
and the corresponding values of the $C,D$ variables:
$$C_{\al,n}=q^{\al(\al-1)\over 2} \, \bR_{\al,n}\bR_{\al,n-1}^{-1} \qquad D_{\al,n}= \bR_{\al+1,n}\bR_{\al,n}^{-1}
\bR_{\al,n-1}^{-1}\bR_{\al-1,n-1}$$
which are all Laurent monomials of the initial data, with possibly some
integer power of $q$ in factor. The Theorem follows immediately from the expressions of
Theorem \ref{ncweights}.
\end{proof}

\end{document}